\documentclass[journal]{IEEEtran}

\usepackage{amsmath}
\usepackage{amssymb,color,graphicx}
\usepackage{subfigure}
\usepackage{cite}
\usepackage{verbatim}
\usepackage{algorithm}
\usepackage{algorithmic}
\usepackage{url}

\newtheorem{lemma}{Lemma}
\newtheorem{proposition}{Proposition}
\newtheorem{definition}{Definition}

\newcommand{\eg}{\emph{e.g., }}
\newcommand{\ie}{\emph{i.e., }}

\begin{document}

\title{Delay Sensitive Communications over Cognitive Radio Networks}

\author{Feng Wang, Jianwei Huang, \emph{IEEE Senior Member}, and Yuping Zhao
\thanks{Feng Wang is with the Beijing Space Technology Development and Test Center,
China Academy of Space Technology, Beijing 100094, China, and was
with the State Key Laboratory of Advanced Optical Communication
Systems \& Networks, Peking University, Beijing 100871, China
(email: fengwangpku@gmail.com). Yuping Zhao is with the State Key
Laboratory of Advanced Optical Communication Systems \& Networks,
Peking University, Beijing 100871, China (email:
yuping.zhao@pku.edu.cn). Jianwei Huang (corresponding author) is
with the Department of Information Engineering, the Chinese
University of Hong Kong, Shatin, Hong Kong (email:
jwhuang@ie.cuhk.edu.hk). Part of the work was done when Feng Wang
visited the Chinese University of Hong Kong between July to
December, 2009. Part of the results was presented in IEEE GLOBECOM
2010 \cite{ACGlobecom10}. The authors acknowledge the contribution
of Dr.~Junhua Zhu in \cite{ACGlobecom10}.

This work is supported by the General Research Funds (Project Number
412710 and 412511) established under the University Grant Committee
of the Hong Kong Special Administrative Region, China, and by Important National Science and Technology Specific
Projects of China (Project Number 2009ZX03003-011-01).} }

\maketitle


\begin{abstract}
Supporting the quality of service of unlicensed users in cognitive radio
networks is very challenging, mainly due to  dynamic  resource availability because of  the licensed users' activities. In this paper, we study the optimal admission control
and channel allocation decisions in cognitive overlay networks in order to support
delay sensitive communications of unlicensed users. We formulate it as a Markov decision
process problem, and solve it by transforming the original
formulation into a stochastic shortest path problem. We then propose
a simple heuristic control policy, which includes a threshold-based admission
control scheme and and a largest-delay-first channel allocation scheme, and prove
the optimality of the largest-delay-first channel allocation scheme. We
further propose an improved policy using the rollout algorithm. By
comparing the performance of both proposed policies with the
upper-bound of the maximum revenue,  we show that our policies
achieve close-to-optimal performance with low complexities.
\end{abstract}

\begin{IEEEkeywords}
Admission control, Markov decision process, Bellman's equation,
rollout algorithm
\end{IEEEkeywords}


\section{Introduction}\label{Sec:Introduction}

Cognitive radio technology has the potential to significantly
improve spectrum utilization and accommodate many more devices in
the limited spectrum. Supporting Quality of Service (QoS), however,
is challenging in cognitive radio networks due to the dynamically
changing network resources. In this paper, we will
design an admission control and channel allocation mechanism to
support delay-sensitive real-time secondary unlicensed communications. Compared
with the resource allocation in conventional communication networks,
the unique challenge here is to incorporate the impact of primary
licensed users on the availability of the communication resources.

Optimal channel selection of a \emph{single} secondary unlicensed
user has been well studied in the literature (\eg
\cite{ZhaoKrLi2008,LiuKrLi2010}). Zhao~\emph{et
al.}~\cite{ZhaoKrLi2008}  considered the  total expected reward
maximization problem when the secondary user can only sense one
channel at a time. Liu \emph{et al.}~\cite{LiuKrLi2010} further
considered the case where the secondary user can sense multiple
channels simultaneously. The resource allocation problem becomes
more complicated when there are multiple secondary users (\eg
\cite{ZhouLiLiWaSo2010,UrN09}). Zhou \emph{et
al.}~\cite{ZhouLiLiWaSo2010} jointly considered channel allocation
with power control. Urgaonkar and Neely \cite{UrN09} developed
opportunistic scheduling policies to provide performance guarantees.

Admission control is critical for supporting QoS when there are too
many users that want to access the network simultaneously. In
traditional cellular networks, many results have shown that the
optimal admission control policy has a threshold structure
(\eg\cite{ChauWoLi2006,HoYP02,KimGrGo2010}). In cognitive radio
networks, researchers have studied admission control for both
underlay networks (\emph{e.g.},
\cite{KangLiGa2010,XiZSH09,LeHo2008}) and overlay networks
(\eg~\cite{KiS09,MuAD09}). In cognitive overlay networks, admission
control is often jointly pursued with channel allocation, as the
secondary users can only access idle channels not occupied by
primary users. Admission control also can be jointly considered with
other mechanisms, \emph{e.g.}, Kim and Shin~\cite{KiS09} considered
joint optimal admission and eviction control using semi-Markov
decision process and linear programming. Mutlu \emph{et al.}
\cite{MuAD09} investigated the problem of optimal spot pricing of
spectrum for maximizing the profit from the admission of secondary
users.

In this paper, we consider the joint admission control and channel
allocation problem for cognitive overlay networks. {Our problem is
very different from the throughput maximization for elastic data
traffic studied in most previous
literature~\cite{ZhouLiLiWaSo2010,UrN09}.} We want to support the
secondary users' real-time applications (\eg VoIP and video streaming) with {stringent delay
constraints}.

The rest of the paper
is organized as follows. We describe the system model in Section
\ref{Sec:SystemModel}, and formulate the admission control and
channel allocation problem as a Markov Decision Process (MDP) in
Section \ref{Sect:ProblemFormualtion}. In Section
\ref{Sec:Analysis}, we transform the problem into a stochastic
shortest path problem and prove the convergence of the Bellman's
equation. Section \ref{Sec:DynamicProgramming} proposes a heuristic
control policy and an improved rollout policy, together with the
corresponding theoretical analysis and simulation results. We
finally conclude in Section \ref{Sec:Conclusion}.


\section{System Model} \label{Sec:SystemModel}

{This paper studies a cognitive radio network as shown in Fig.
\ref{fig:SystemModel}. We consider an infrastructure-based secondary
unlicensed network, where a secondary network operator senses the
channel availabilities (\ie primary licensed users' activities) and
decides the admission control and channel allocation for the
secondary users. A similar network architecture has been  considered
in several recent literature (\eg
\cite{Lehr07,Peha09,DuHS10-1,DuHS10-2}).
Comparing with the distributed network architecture where end users
need to perform spectrum sensing individually, the network
architecture considered in this paper has the advantage of reducing
the complexity of the secondary user devices and providing better
QoS support. Such infrastructure-based network without user sensing
requirement is also consistent with the recent ruling of FCC (Federal
Communications Commission) on the TV white space sharing \cite{fcc}.}

\begin{figure}
\centering
  \includegraphics[width=1.0\columnwidth]{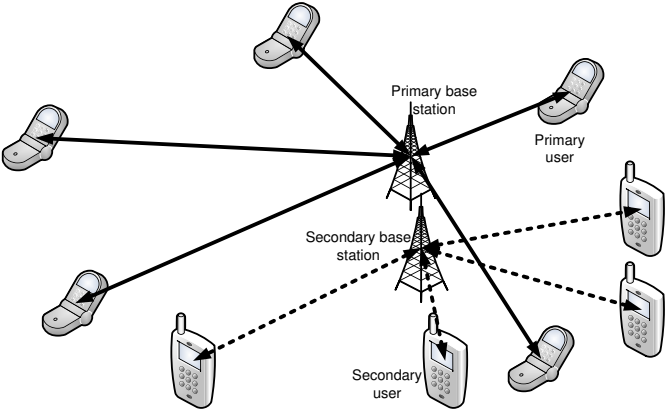}\\
  \caption{A cognitive radio network scenario. In the secondary network,
  the dotted arrows denote the channels between the secondary
  base station and the secondary users.}\label{fig:SystemModel}
\end{figure}

{One way to realize network-based spectrum sensing is to construct a
sensor network that is dedicated to sensing the radio environment in
space and time~\cite{Weiss2010.sensing}. The secondary base station
will collect the sensing information from the sensor network and
provide it to the unlicensed users, which is called ``sensing as
service''. There has been significant current research efforts along
this direction in the context of an European project
SENDORA~\cite{SENDORA}, which aims at developing techniques based on
sensor networks for supporting coexistence of licensed and
unlicensed wireless users in a same area.}

{In our model, the time is divided into equal length slots. Primary
users' activities remain roughly unchanged within a single time
slot. This means that it is enough for the operator to sense once at
the beginning of each time slot (see Fig. \ref{fig:timeslot}). For
readers who are interested in the optimization of the time slot
length to balance sensing and data transmission, see
\cite{Senhua2009.sensingtransmission}.}

\begin{figure}
\centering
\includegraphics[width=1.0\columnwidth]{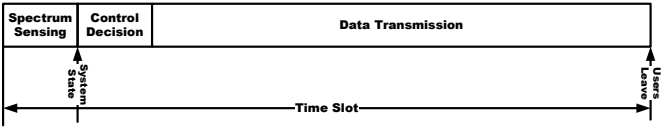}\\
\caption{The components of a time slot.}\label{fig:timeslot}
\end{figure}

\begin{figure}
\centering
\includegraphics[width=1.0\columnwidth]{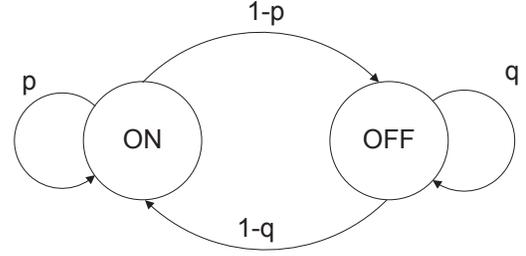}\\
\caption{Markovian ON/OFF model of channel
activities.}\label{fig:OnOff}
\end{figure}

The network has a set $\mathcal{J}=\{1,\ldots,J\}$ of orthogonal
primary licensed channels. The state of each channel follows a
Markovian ON/OFF process as in Fig.~\ref{fig:OnOff}. {If a channel
is ``ON'', then it means that the primary user is not active on the
channel and the channel condition is good enough to support the
transmission rate requirement of a secondary user. Here we assume
that all secondary users want to achieve the same target
transmission rate (\eg that of a same type of video streaming
application). If a channel is ``OFF'', then either a  primary user
is active on this channel, or the channel condition is not good
enough to achieve the secondary user's target rate. In the time
slotted system,} the channel state changes from ``ON'' to ``OFF''
(``OFF'' to ``ON'', respectively) between adjacent time slots with a
probability $p$ ($q$, respectively). When a channel is ``ON'', it
can be used by a secondary unlicensed user.

We consider an infinitely backlog case, where there are many
secondary users who want to access the idle channels. Each idle
channel can be used by at most one secondary user at any given time.
A secondary user represents an unlicensed user communicating with
the secondary base station as shown in Fig. \ref{fig:SystemModel}.
The secondary users are interested in real-time applications such as
video streaming and VoIP, which require steady data rates with
{stringent delay constraints. The key QoS parameter is the
\emph{accumulative delay}, which is the total delay that a secondary
user experiences after it is admitted into the system.} Once a
secondary user is admitted into the network, it may finish the
session \emph{normally} with a certain probability. However, if the
user experiences an accumulative delay larger than a threshold, then
its QoS significantly drops (\eg freezing happens for video
streaming) and the user will be \emph{forced to terminate}.

To make the analysis {tractable}, we make several assumptions.
First, we assume that the availabilities of all channels follow the
same Markovian model. This is reasonable if the traffic types of
different primary users are similar (\eg all primary users are voice
users).  Second, we assume that all secondary users experience the
same channel availability independent of their locations. This is
reasonable when the secondary users are close-by. Third, we assume
the spectrum sensing is error-free. This can be well approximated by
having enough sensors performing collaborating sensing. Furthermore,
we assume that all channels are homogeneous and can provide the same
data rate to any single secondary user using any channel. Finally,
we assume that all secondary users are homogeneous (\ie interested
in the same application such as video streaming). Each secondary
user only requires one available channel to satisfy its rate
requirement. Several of the above assumptions can be relaxed by
increasing the state space of the MDP formulation. As we will see
shortly, the admission control and channel allocation issue in this
homogeneous case is already quite complicated and admits no
closed-form solutions. The analysis and insights of this paper will
enable us to further consider heterogeneous channels and secondary
users in the future.


\section{Problem Formulation}\label{Sect:ProblemFormualtion}

We formulate the admission control and channel
allocation problem {as an MDP}~\cite{Ber05}.
In an infinite-horizon MDP with a set of
finite states $\mathcal{S}$, the state evolves through time
according to a transition probability matrix
{$\left\{P_{x_kx_{k+1}}\right\}$}, which depends on both
the current state and the control decision from a set $\mathcal{U}$.
More specifically, if the network is in state {$x_k$} in time
slot $k$ and selects a decision {$u(x_k) \in \mathcal{U}(x_k)$},
then the network obtains a revenue $g(x_k,u(x_k))$ in time slot $k$
and moves to state {$x_{k+1}$} in time slot $k+1$ with
probability {$P_{x_kx_{k+1}}(u(x_k))$}. We want to maximize
the long-term time average revenue, \ie
\begin{equation}
\label{equ:RevenueFunction}
    \lim_{T \to
    \infty}E\left\{{\frac{1}{T}\sum_{k=0}^{T-1}{g(x_k,u(x_k))}}\right\}.
\end{equation}


\subsection{The State Space}\label{Sec:sub.StateSpace}

{The system state describes system information after the network
performs  spectrum sensing at the beginning of the time slot (see
Fig. \ref{fig:timeslot})}. It consists of two components:
\begin{itemize}
  \item A \emph{channel state component},
  $m=\boldsymbol{a}^T \cdot \boldsymbol{a}$, describes the number of available channels. Here $\boldsymbol{a}=(a_j,\forall j\in\mathcal{J})$ is the  channel availability vector,
  where $a_j=1$ (or $0$) when channel $j$ is available (or not).

\item A \emph{user state component}, $\boldsymbol{\omega_e}=(\omega_{e,i},\forall i\in\mathcal{D})$,
describes the numbers of secondary users with different accumulative
delays. Here $\mathcal{D}=\{0,1,\ldots,D_{\max}\}$ is the set of
possible delays, and $\omega_{e,i}$ denotes the number of secondary
users whose accumulative delay is $i$.
\end{itemize}

We let $\mathcal{M}$ denote the feasible set of the channel state
component, and $\Omega$ denote the feasible set of the user state
component. The state space is given by
{$\mathcal{S}=\left\{(m,\boldsymbol{\omega_e)}|m \in
\mathcal{M},\boldsymbol{\omega_e} \in \Omega\right\}.$}

State $\theta$ is said to be \emph{accessible} from state $\eta$ if
and only if  it is possible to reach state $\theta$ from $\eta$, \ie
$P\{reach~\theta|start~in~\eta\}>0$ \cite{Ross07-b}. Two states that
are accessible to each other are said to be able to
\emph{communicate} with each other. In our formulation, all the
states in space $\mathcal{S}$ are accessible from state
$\boldsymbol{0}$, {which is defined as  a state where there is no
available channel and no single admitted secondary user in the
system}. Since it is possible to have $m=0$ in several consecutive
time slots (when primary traffic is heavy and occupies all
channels), thus state $\boldsymbol{0}$ is accessible from any state
in the state space $\mathcal{S}$. Hence, all the states communicate
with each other and the Markov chain is \emph{irreducible}. Finally,
the state space is finite, so all the states are \emph{positive
recurrent} \cite{Ross07-b}. This property turns out to be critical
for the analysis in Section \ref{Sec:Analysis}.


\subsection{The Control Space}

{For the state $x_k=\{m,\boldsymbol{\omega_e}\} \in \mathcal{S}$ in
each time slot $k$,} the set of available control choices
{$\mathcal{U}(x_k)$} depends on the relationship between the channel
state and the user state. The control vector {$u(x_k)=\{u_a,
\boldsymbol{u_e}\}$} consists of two parts: scalar $u_a$ denotes the
number of admitted new secondary users, and vector
$\boldsymbol{u_e}=\{u_{e,i}, \forall i\in\mathcal{D} \}$ denotes the
numbers of secondary users who are allocated channels and have
accumulative delays of $i\in\mathcal{D}$ at the beginning of the
current time slot.
Without loss of generality, we assume $0 \leq u_a \leq J$, \ie we
will never admit more secondary users than the total number of
channels. This leads to $0 \leq u_{e, 0} \leq \omega_{e, 0} + u_a$,
$0 \leq u_{e, i} \leq \omega_{e, i}$ for all $ i\in [1,D_{\max}]$,
and $0 \leq \sum_{i=0}^{D_{\max}}{u_{e,i}} \leq m$. Since {$m \leq
J$}, the cardinality of the control space $\mathcal{U}$ is
{$J^{D_{\max}+2}$.}


\subsection{The State Transition}

Current state {$x_k=\{m, \boldsymbol{\omega_e}\} \in \mathcal{S}$}
together with the control {$u(x_k) \in \mathcal{U}(x_k)$}
determine the probability of reaching the next state
{$x_{k+1}=\{m^\prime, \boldsymbol{\omega_e^\prime}\}$}.

First, the transition of channel state component from $m$ to
$m^{\prime}$ depends on the underlying primary traffic. We can
divide $m^\prime$ available channels into two groups:  one group
contains $m_1^\prime$ channels which are available {in the (current) time
slot $k$}, the other group contains $m_2^\prime$ channels which are not
available {in time slot $k$}. Let us define the set
$\mathcal{Z}\!\!=\!\!\left\{(m_{1}^{\prime},m_{2}^{\prime})|
m^\prime\!\!=\!\!m_1^\prime\!\!+\!\!m_2^\prime, 0 \!\!\leq\!\!
m_1^\prime \!\!\leq\!\! m, 0 \!\!\leq\!\! m_2^\prime \!\!\leq\!\!
J\!\!-\!\!m\right\}.$ Then we can calculate the probability based
on the {i.i.d.} ON/OFF model in Section \ref{Sec:SystemModel}:
\begin{equation}\label{equ:TransitionProbabilitymm}
P_{mm^\prime}\!\!=\!\!\!\!\sum_{(m_{1}^{\prime},m_{2}^{\prime})
\in\mathcal{Z}}\!\!\left\{\!\!\binom{m}{m_1^\prime}\!\!p^{m_1^\prime}(1\!\!-\!\!p)^{m\!\!-\!\!m_1^\prime}\!\!\binom{J\!\!-\!\!m}{m_2^\prime}\!\!(1\!\!-\!\!q)^{m_2^\prime}q^{J\!\!-\!\!m\!\!-\!\!m_2^\prime}\right\}.
\end{equation}
Thus the channel transition function is $f_s(m) = m^\prime$
with probability $P_{mm^\prime}$ for all $m'\in{\mathcal{M}}$.

Let us define {$\boldsymbol{\omega_c}=\{\omega_{c,i},\forall
i\in\mathcal{D}\}$} as the number of secondary users who normally
complete their connections (not due to delay violation) {in time
slot $k$}. For example, a user may terminate a video streaming
session after the movie finishes, or terminate a VoIP session when
the conversation is over. If we assume that all users have the same
completion probability $P_f$ per slot when they are actively served,
then the event of having $\rho$ out of $\tau$ users completing their
connections (denoted as $f_c(\tau)= \rho$) happens with probability
{$\binom{\tau}{\rho}{P_f^\rho(1-P_f)^{\tau-\rho}}$}.

Finally, define $\omega_q$ as the number of secondary users who are
forced to terminate their connections {during time slot $k$}. The
state transition can be written as
{
\begin{equation}\label{equ:StatesTransation}
    \left\{\begin{split}
        &m^\prime=f_s(m), \\
        &\omega_{c,i} = f_c(u_{e,i}), \forall i\in\mathcal{D}, \\
        &\omega_q=\omega_{e,D_{\max}}-u_{e,D_{\max}}, \\
        &\omega_{e,0}^\prime=u_{e,0}-\omega_{c,0}, \\
        &\omega_{e,1}^\prime\!=\!u_{e,1} + (\omega_{e,0}+u_a-u_{e,0})-\omega_{c,1}, \\
        &\omega_{e,i}^\prime\!=\!u_{e,i}\! + \!(\omega_{e,i-1}-u_{e,i-1}) \!-\! \omega_{c,i}, \forall i \in [2,D_{\max}].
        \end{split}\right.
\end{equation}}

{Let us take a network with $J=10$ and $D_{max}=2$ as a numerical
example. In a particular time slot, assume that there are $m=7$ channels
available and a total of $6$ secondary users admitted in the system:
$1$ user with zero accumulative delay, $3$ users with $1$ time slot
of accumulative delay, and $2$ users with $2$ time slots of
accumulative delay. Then the state vector is $\{m,
\boldsymbol{\omega_e}\}=\{7,(1,3,2)\}$. Assume the control decision
is to admit $2$ new users and to allocate available channels to the
users except one of the new users, \ie $u=\{u_a,
\boldsymbol{u_e}\}=\{2,(2,3,2)\}$. Thus if there is no user
completing a connection in the current time slot and $m^\prime=4$
available channels in the next time slot, the system state becomes $\{m^\prime, \boldsymbol{\omega_e^\prime}\}=\{4,(2,4,2)\}$.}


\subsection{The Objective Function}

Our system optimization objective is to choose the optimal control decision for each
possible state to maximize the expected average revenue per time
slot (also called stage), \ie
\begin{equation}\label{equ:RevenueFunction-2}
    \max{\lim_{T \to
    \infty}E\left\{{\frac{1}{T}\sum_{k=0}^{T-1}{g(x_k,u(x_k))}}\right\}}.
\end{equation}
Here the revenue function is computed at the end of each time slot
$k$ as follows:
\begin{equation}\label{equ:RevenueFunction-3}
    g(x_k,u(x_k))\!\!=\!\! R_c\!\!\sum_{i=0}^{D_{\max}}{\omega_{c,i}(k)} \!\!+\!\! R_t\!\!\sum_{i=0}^{D_{\max}}{\omega_{e,i}}(k) \!\!-\!\! C_q\omega_q(k),
\end{equation}
where $R_c \geq 0$ is the reward of completing the connection of a
secondary user normally (without violating the maximum delay
constraints), $R_t \geq 0$ is the reward of maintaining the
connection of a secondary user, and $C_q  \geq 0$ is the penalty of
forcing to terminate a connection. By choosing different values of
$R_{c}$, $R_{t}$, and $C_{q}$, a network designer can achieve
different objective functions. In this paper, we assume that the
values of $R_{c}$, $R_{t}$, and $C_{q}$ are given parameters.


\section{Analysis of the MDP Problem}\label{Sec:Analysis}

We define a sequence of control actions as a policy,
{${\boldsymbol{\mu}}=\left\{u(x_0),u(x_1),\cdots\right\}$, where
$u(x_k)\in\mathcal{U}(x_k)$ for all $k$. A policy
is stationary if the choice of decision
only depends on the state and is independent of the time.}
Let
$$V_{\boldsymbol{\mu}}(\theta)=\lim_{T \to
\infty}E\left\{\frac{1}{T}\sum_{k=0}^{T-1}g(x_k,u(x_k))|x_0=\theta\right\}$$
be the expected revenue in state
$\theta$ under policy $\boldsymbol{\mu}$. Our objective is to find
the best policy $\boldsymbol{\mu^\ast}$ to optimize the average revenue
per stage starting from an initial state $\theta$.

Section \ref{Sec:sub.StateSpace} shows that any state can be visited
from any other state within finite stages under a stationary
policy.\footnote{A policy is stationary if the choice of decision
only depends on the state and is independent of the time.} Moreover,
since the revenue $g(x_k,u(x_k)) < \infty$ for all $x_k$ and $u$, we
have
\begin{equation}\label{equ:IndependentOfInitialState}
\lim_{T \to
\infty}\frac{1}{T}E\left\{\sum_{k=0}^{K}g(x_k,u(x_k))\right\}=0
\end{equation}
for any finite $K$. Therefore, we have the following proposition in
our prior preliminary results \cite{ACGlobecom10}.
\begin{proposition}\label{pro:InitialStateInpendent}
For any stationary policy, the average revenue per stage is
independent of the initial state.
\end{proposition}

Next we give the following detailed proof of the proposition.
\begin{proof}
Since the revenue $g(x_k,u(x_k)) < \infty$ for all $x_k$ and $u$, we
have
\begin{equation}\label{equ:IndependentOfInitialState}
\lim_{T \to
\infty}\frac{1}{T}E\left\{\sum_{k=0}^{K}g(x_k,u(x_k))\right\}=0
\end{equation}
for any finite value of $K$. Consider a stationary policy
$\boldsymbol{\mu}$ whose control decision only depends on the state
of the system. According to the MDP formulation, all the states are
positive recurrent. So starting in state $\theta$, the process will
visit state $\eta$ infinitely often; and the expected time that the
process visits state $\eta$ from state $\theta$ is finite
\cite{Ross07-b}. Thus, any state in the state space can be visited
from any other state within enough stages (finite) under the
stationary policy. Therefore, we assume, under the policy
$\boldsymbol{\mu}$, the state $\eta \in \mathcal{S}$ is visited from
the state $\theta \in \mathcal{S}$. Let
$K_{\theta\eta}(\boldsymbol{\mu})$ be the number of time slots that
the system first passes state $\eta$ from state $\theta$ under
policy $\boldsymbol{\mu}$, then the average revenue per stage
corresponding to initial condition $x_0=\theta$ can be expressed as
\begin{equation}\label{equ:IndependentOfInit-0}
\begin{split}
V_{\boldsymbol{\mu}}(\theta)=&\lim_{T \to
\infty}\frac{1}{T}E\left\{\sum_{k=0}^{K_{\theta\eta}(\boldsymbol{\mu})-1}g(x_k,u(x_k))\right\}\\
&+\lim_{T \to
\infty}\frac{1}{T}E\left\{\sum_{k=K_{\theta\eta}(\boldsymbol{\mu})}^{T-1}g(x_k,u(x_k))\right\}.
\end{split}
\end{equation}
The first term in \eqref{equ:IndependentOfInit-0} is zero according
to \eqref{equ:IndependentOfInitialState}, while the second limit is
equal to $V_{\boldsymbol{\mu}}(\eta)$. So with
$E\{K_{\theta\eta}(\boldsymbol{\mu})\}<\infty$,
\begin{equation}\label{equ:IndependentOfInit}
    V_{\boldsymbol{\mu}}(\theta)=V_{\boldsymbol{\mu}}(\eta)=V_{\boldsymbol{\mu}},
\end{equation}
for any two states $\theta$ and $\eta$.
\end{proof}

As shown in Proposition \ref{pro:InitialStateInpendent}, the average
revenue per stage under any stationary policy is independent of the
initial state, and the \emph{average revenue maximization problem}
could be transformed into the \emph{stochastic shortest path
problem}. {More specifically, we pick a state $n$ as the start state
of the stochastic shortest path problem, and define an
\emph{artificial} termination state $t$ from the state $n$. The
transition probability from an arbitrary state $\theta$ to the
termination state $t$ satisfies $P_{\theta
t}(\boldsymbol{\mu})=P_{\theta n}(\boldsymbol{\mu})$, as show in
Fig. \ref{fig:ShortestPathProblem}.}

\begin{figure}
\centering
\includegraphics[width=1.0\columnwidth]{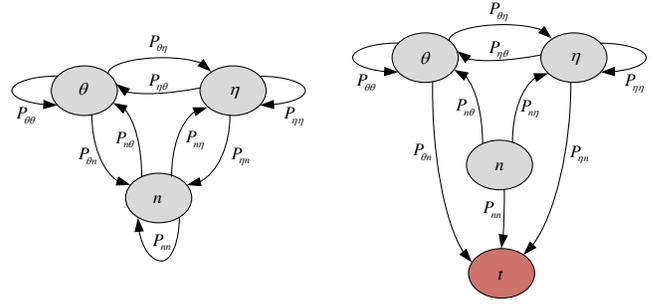}\\
\caption{Transition probability of the shortest path
problem.}\label{fig:ShortestPathProblem}
\end{figure}

In the stochastic shortest path problem, we define
$-\hat{g}(n,\boldsymbol{\mu})$ as the expected stage \emph{cost}
incurred at state $n$ under policy $\boldsymbol{\mu}$.
Let $A^\ast$ be the optimal average revenue per stage starting from
the state $n$ to the terminal state $t$, and let
$A^\ast-\hat{g}(n,\boldsymbol{\mu})$ be the normalized expected
stage cost. Then the normalized expected terminal cost from the
state $x_0=n$ under the policy $\boldsymbol{\mu}$,
$h^{\boldsymbol{\mu}}(n)=\lim_{N \to
\infty}{E\left\{\sum_{k=0}^{N-1}\left\{A^\ast-g(x_k,u(x_k))\right\}\right\}}$,
is zero when the policy $\boldsymbol{\mu}$ is optimal. The cost
minimization in the stochastic shortest path problem is equivalent
to the original average revenue per stage maximization problem. Let
$h^\ast(\theta)$ denote the optimal cost of the stochastic shortest
path starting at state $\theta \in \mathcal{S}$, then we get the
corresponding Bellman's equation as follows \cite{Ber05}:
\begin{equation}\label{equ:BellmanEquationForShortestPath-1}
h^{\boldsymbol{\mu}}(\theta)\!\!=\!\!\min_{\boldsymbol{\mu}}\left\{A^\ast\!\!-\!\!\hat{g}(\theta,\boldsymbol{\mu})\!\!+\!\!\sum_{\eta\in\mathcal{S}}p_{\theta\eta}(\boldsymbol{\mu})h^{\boldsymbol{\mu}}(\eta)\right\},~\theta\in\mathcal{S}.
\end{equation}
If $\boldsymbol{\boldsymbol{\mu}^\ast}$ is a stationary policy that
maximizes the cycle revenue, we have the following equations:
\begin{equation}\label{equ:BellmanEquationForShortestPath-3}
h^\ast(\theta)=A^\ast-\hat{g}(\theta,\boldsymbol{\mu}^\ast)+\sum_{\eta\in\mathcal{S}}p_{\theta\eta}(\boldsymbol{\mu}^\ast)h^\ast(\eta),~\theta\in\mathcal{S}.
\end{equation}
The Bellman's equation is an iterative way to solve MDP problems.
Next we show that solving the Bellman's equation
(\ref{equ:BellmanEquationShortestPath}) in the stochastic shortest
path problem leads to the optimal solution.

\begin{proposition}\label{pro:Convergence}
For the stochastic shortest path problem, given any initial
values of terminal costs $h_0(\theta)$ for all states $\theta\in\mathcal{S}$, the sequence
$\{h_l(\theta),l=1,2,\ldots\}$ generated by the iteration
\begin{equation}\label{equ:BellmanEquationShortestPath}
h_{l+1}(\theta)\!\!=\!\!\min_{\boldsymbol{\mu}}{\left\{A^\ast\!\!-\!\!\hat{g}(\theta,\boldsymbol{\mu})\!\!+\!\!\sum_{\eta\in\mathcal{S}}P_{\theta\eta}(\boldsymbol{\mu})h_l(\eta)\right\},
{\theta\in\mathcal{S},}}
\end{equation}
converges to the optimal terminal cost $h^\ast(\theta)$ for each state
$\theta$.
\end{proposition}

\begin{proof}
For an arbitrary state $\theta$\ and an admissible policy $\boldsymbol{\mu}$,
there exists an integer $\gamma$ satisfying
$P\{x_\gamma \neq t|x_0=\theta,\boldsymbol{\mu}\}<1$ \cite{Ber05-b}.
Let $\rho=\max_{(\theta, \boldsymbol{\mu})}{P\{x_\gamma \neq
t|x_0=\theta,\boldsymbol{\mu}\}}$, then $\rho<1$ and $
P\left\{x_{2\gamma}\neq t|x_0=\theta,\boldsymbol{\mu}\right\}
=P\{x_{2\gamma}\neq t|x_{\gamma}\neq t, x_0=\theta,
\boldsymbol{\mu}\} \cdot P\{x_{\gamma}\neq
t|x_0=\theta,\boldsymbol{\mu}\} \leq \rho^2.$
Therefore, we get
$P\{x_{\phi\gamma} \neq t | x_0=\theta, \boldsymbol{\mu}\} \leq
\rho^\phi.$

We break down the cost $h^{\boldsymbol{\mu}}(x_0)$ into the portions
incurred over the first $K\gamma$ time slots ($K$ is a positive
integer) and over the remaining time slots,
\ie
\begin{equation}\label{TheCostBrokeDown1}
    \begin{split}
        h^{\boldsymbol{\mu}}(x_0)=&\lim_{N \to \infty}{E\left\{\sum_{k=0}^{N-1}\left\{A^\ast-g(x_k,u(x_k))\right\}\right\}}\\
        =& E\left\{\sum_{k=0}^{K\gamma-1}\left\{A^\ast-g(x_k,u(x_k))\right\}\right\} \\
        &+\lim_{N \to \infty}{E\left\{\sum_{k=K\gamma}^{N-1}\left\{A^\ast-g(x_k,u(x_k))\right\}\right\}}.
    \end{split}
\end{equation}
Define $\Gamma=\gamma
\max_{(\theta,\boldsymbol{\mu})}{\left|A^\ast-\hat{g}(\theta,\boldsymbol{\mu})\right|}$,
which denotes the upper bound on the cost of an $\gamma$-slot cycle
when termination does not occur during the cycle. Then, the expected
cost during the $K$-th $\gamma$-slot cycle (time slots $K\gamma$ to
$(K+1)\gamma-1$) is upper bounded by $\rho^K \Gamma$, so that
\begin{equation}\label{equ:CostBound1}
\begin{split}
&E\left\{\left|h^{\boldsymbol{\mu}}(x_0)-\sum_{k=0}^{K\gamma-1}\{A^\ast-g(x_k,u(x_k))\}\right|\right\}\\
&= \left|\lim_{N \to
\infty}{E\left\{\sum_{k=K\gamma}^{N-1}\{A^\ast-g(x_k,u(x_k))\}\right\}}\right|\\
&\leq \Gamma\sum_{\phi=K}^\infty{\rho^\phi}=\frac{\rho^K
\Gamma}{1-\rho}.
\end{split}
\end{equation}
Let $h_0(x_0)$ be a terminal cost function as defined in the
proposition, and then its expected value under $\boldsymbol{\mu}$
after $K\gamma$ time slots is bounded by
\begin{equation}\label{TheCostBrokeDown2}
\begin{split}
    \left|E\{h_0(x_{K\gamma})\}\right|&\!\!=\!\!\left|\sum_{\theta\in\mathcal{S}}{P(x_{K\gamma}=\theta|x_0,\boldsymbol{\mu})h_0(\theta)}\right|\\
    &\!\!\leq\!\! \left(\sum_{\theta\in\mathcal{S}}{P(x_{K\gamma}=\theta|x_0,\boldsymbol{\mu})}\right) \!\! \max_{\theta\in\mathcal{S}}{|h_0(\theta)|}.
\end{split}
\end{equation}
Since the probability that $x_{K\gamma} \neq t$ is less than or
equal to $\rho^K$ for any policy, we have
$\left|E\{h_0(x_{K\gamma})\}\right| \leq \rho^K
\max_{\theta\in\mathcal{S}}{|h_0(\theta)|}$. Therefore, we can get
\begin{equation}\label{TheCostBrokeDown3}
    \begin{split}
        &-\rho^K \max_{\theta\in\mathcal{S}}{|h_0(\theta)|} + h^{\boldsymbol{\mu}}(x_0) - \frac{\rho^K
        \Gamma}{1-\rho}\\
        &\leq E\left\{h_0(x_{K\gamma}) +\sum_{k=0}^{K\gamma-1}\{A^\ast-g(x_k,u(x_k))\}\right\}\\
        &\leq \rho^K \max_{\theta\in\mathcal{S}}{|h_0(\theta)|} + h^{\boldsymbol{\mu}}(x_0) + \frac{\rho^K \Gamma}{1-\rho}.
    \end{split}
\end{equation}
The expected value in the middle term of the above inequalities is
the $K\gamma$-slot cost of policy $\boldsymbol{\mu}$ starting from
state $x_0$ with a terminal cost $h_0(x_{K\gamma})$. The minimum of
this cost over all $\boldsymbol{\mu}$ is equal to the value
$h_{K\gamma}(x_0)$, which is generated by the dynamic programming
recursion (\ref{equ:BellmanEquationShortestPath}) after $K\gamma$
iterations. Thus, by taking the minimum over $\boldsymbol{\mu}$ in
\eqref{TheCostBrokeDown3}, we obtain for all $x_0$ and $K$,
\begin{equation}\label{TheCostBrokeDown4}
\begin{split}
        -\rho^K \max_{\theta\in\mathcal{S}}{|h_0(\theta)|} + h^\ast(x_0) - \frac{\rho^K \Gamma}{1-\rho}
        \leq h_{K\gamma}(x_0)&\\
        \leq \rho^K \max_{\theta\in\mathcal{S}}{|h_0(\theta)|} + h^\ast(x_0) + \frac{\rho^K
        \Gamma}{1-\rho}.&
\end{split}
\end{equation}
And by taking the limit when $K \to \infty$, the terms involving
$\rho^{K}$ will go to zero, and we obtain $\lim_{K \to
\infty}{h_{K\gamma}(x_0)}=h^\ast(x_0)$ for all $x_0$. In addition,
since $\left|h_{K\gamma+q}(x_0)-h_{K\gamma}(x_0)\right| \leq
\rho^K\Gamma,\; q=1,2,\cdots,\gamma-1$, we have $\lim_{K \to
\infty}{h_{K\gamma+q}(x_0)} = \lim_{K \to \infty}{h_{K\gamma}(x_0)}
= h^\ast(x_0)$ for all $q=1,\cdots,\gamma-1$.
\end{proof}

{Proposition \ref{pro:Convergence} shows that solving the Bellman's
equation leads to the optimal average revenue $A^\ast$ and the
optimal differential cost $h^\ast$. The Bellman's equation can often
be solved using  value iteration or policy iteration algorithms;
details can be found in \cite{Ber05-b} and \cite{Ber07-b}.} Once
having $A^\ast$ and $h^\ast$, we can compute the optimal control
decision {$u^\ast(\theta)$}  that minimizes the immediate
differential cost of the current stage plus the remaining expected
differential cost for state $\theta$, \ie
{\begin{equation}\label{equ:OptimalControlSolution}
u^\ast(\theta)\!=\!\arg{\min_{\boldsymbol{\mu}}{\left\{A^\ast\!-\!\hat{g}(\theta,\boldsymbol{\mu})\!+\!\sum_{\eta\in\mathcal{S}}{p_{\theta\eta}(\boldsymbol{\mu})h^\ast(\eta)}\right\}}}.
\end{equation}}


\section{Suboptimal Control and Dynamic Programming}\label{Sec:DynamicProgramming}

Solving the Bellman's equation does not lead to a closed-form
optimal control policy, and the iterative computation is
time-consuming for our problem with a large state space. To resolve
this issue, a broad class of suboptimal control methods referred as
\emph{approximate dynamic programming} (ADP) have been proposed in
\cite{Ber05}. Next we first propose a simple heuristic control
policy in Section \ref{Sec:DP.solution}. Then in Section
\ref{Sec:DP.optimal}, we will improve the performance of the
heuristic algorithm by using the idea of \emph{rollout algorithm}
(which is a class of ADP algorithms). It is known that the
suboptimal policy based on the rollout algorithm is identical to the
policy obtained by a \emph{single policy improvement step} of the
classical policy iteration method \cite{Ber05-b,Ber07-b}.


\subsection{Heuristic Control Policy}\label{Sec:DP.solution}

Several observations can help us with the suboptimal algorithm
design. First, the channel state transitions are determined by the
underlying primary traffic and are not affected by any control policy.
Second, all secondary users experience the same channel availability
independent of their locations, and all channels are homogenous and
provide the same data rates. This means that we are interested in
\emph{how many} users to admit rather than \emph{who} to admit, and
we only care \emph{how many} channels are available instead of
\emph{which} are available. This motivates us to first consider
admission control and channel allocation separately.

{For the admission control, we first consider a simple
\emph{threshold-based} strategy, where a new user will be admitted
if and only if the total number of admitted users is smaller than
the threshold. Given a fixed admission control threshold $T_{h}$,}
there are many ways of performing the channel allocation. To resolve
this issue, we propose the \emph{largest-delay-first} strategy,
which allocates available channels to admitted users with the
largest accumulated delay first.

\begin{proposition}\label{pro:OptimalChannelAllocation}
The largest-delay-first channel allocation policy is optimal under
any fixed threshold-based admission control policy.
\end{proposition}
\begin{proof}\label{proof:OptimalChannelAllocation} 
Under a threshold-based admission control policy, the number of
admitted users in the system is constant in any time slot. The
objective function in (\ref{equ:RevenueFunction-2}) is equal to the
maximization problem $\max{E\left\{g(x,u(x))\right\}}$ due to
ergodicity of the instant revenue $g(x_k,u(x_k))$. Let
$\Omega_c=E\left\{\sum_{i=0}^{D_{max}}{\omega_{c,i}}\right\}$ be the
expected number of normally completed users at the end of each time slot,
$\Omega_e=E\left\{\sum_{i=0}^{D_{max}}{\omega_{e,i}}\right\}$ be the
expected number of users in the network at the end of each time slot, and
$\Omega_q=E\left\{\omega_q\right\}$ be the expected number of
forcefully terminated users at the end of each time slot. Then
\begin{equation}
\max{E\left\{g(x,u(x))\right\}}=\max{\left\{R_c\Omega_c+R_t\Omega_e-C_q\Omega_q\right\}}.
\end{equation}
Under the threshold-based admission control policy,
{$\Omega_e=\sum_{i=0}^{D_{max}}{\omega_{e,i}}$} in all time slot
$k$ and equals to the threshold.

In the largest-delay-first policy, let $L_c$ be the expected length
of a normally completed session, $D_c$ the expected delay of a
normally completed session, and $L_q$ the expected length of a
forcefully terminated session. Now let us consider an arbitrary
channel allocation policy as the benchmark, and we use the
superscript ${(g)}$ to denote all parameters corresponding to this
particular channel allocation policy, \ie $\Omega_c^{(g)}$,
$\Omega_e^{(g)}$, $\Omega_q^{(g)}$, $L_c^{(g)}$, $D_c^{(g)}$, and
$L_q^{(g)}$. We will show that the largest-delay-first policy is no
worse than this benchmark policy, which will prove the proposition.

Because all actively served users have the same completion
probability $P_f$ independent of the channel allocation decisions,
we can show that $\Omega_c=\Omega_c^{(g)},$
$\Omega_e=\Omega_e^{(g)}$, and  $L_c = L_c^{(g)}$. Since the
largest-delay-first policy always allocates available channels to
the secondary users with the largest delay, we have $D_c \geq
D_c^{(g)}.$

Here comes the critical proof step. We consider $\Omega_e$ virtual
channels, one for each user in the network. If the secondary user is
allocated an available \emph{physical} channel, then its virtual
channel is ``idle'' in that time slot; otherwise its virtual channel
is  ``busy'' and causes a delay. In the long run (when $T \to
\infty$), we have the following:
\begin{equation}\label{equ:EquationUsedResource}
\begin{split}
\Omega_e\cdot T&=\Omega_c T(L_c+D_c)+\Omega_q T(L_q+D_{max})\\
&=\Omega_c^{(g)} T(L_c^{(g)}+D_c^{(g)})+\Omega_q^{(g)}
T(L_q^{(g)}+D_{max}).
\end{split}
\end{equation}
Based on the relationships we just derived in the previous
paragraph, we have
\begin{equation}\label{equ:LessEqual}
\Omega_q(L_q+D_{max}) \leq \Omega_q^{(g)}(L_q^{(g)}+D_{max}).
\end{equation}
Since the number of available channels is the
same under the two channel allocation policies in any time slot , we have
\begin{equation}\label{equ:EquationAvailableChannelsNum}
\Omega_c T L_c + \Omega_q T L_q = \Omega_c^{(g)} T L_c^{(g)} +
\Omega_q^{(g)} T L_q^{(g)}.
\end{equation}
Since $\Omega_c=\Omega_c^{(g)}$ and $L_c = L_c^{(g)}$,
(\ref{equ:EquationAvailableChannelsNum}) implies that $\Omega_q L_q
=\Omega_q^{(g)} L_q^{(g)}$. Together with inequality
(\ref{equ:LessEqual}), we have $\Omega_q \leq \Omega_q^{(g)}.$

Because $\Omega_c=\Omega_c^{(g)},$ $\Omega_e=\Omega_e^{(g)}$, and
$\Omega_q \leq \Omega_q^{(g)}$, we have
$R_c\Omega_c+R_t\Omega_e-C_q\Omega_q \geq
R_c\Omega_c^{(g)}+R_t\Omega_e^{(g)}-C_q\Omega_q^{(g)},$ \ie
$\max{E\left\{g(x_k,u(x_k))\right\}} \geq
\max{E\left\{g^{(g)}(x_k,u(x_k))\right\}}.$ This shows that our
proposed largest-delay-first channel allocation policy is no worse
than any channel allocation algorithm, and thus is optimal with a
threshold-based admission control.
\end{proof}

For performance comparison, we further define two benchmark
channel allocation strategies.
\begin{itemize}
  \item \textbf{Strategy 1:} allocate the available channels to the admitted users with the smallest accumulated delays. If there is a tie, break it randomly.
  \item \textbf{Strategy 2:} allocate the available channels to the admitted users randomly.
\end{itemize}

\begin{figure}
\centering
\includegraphics[width=1.0\columnwidth]{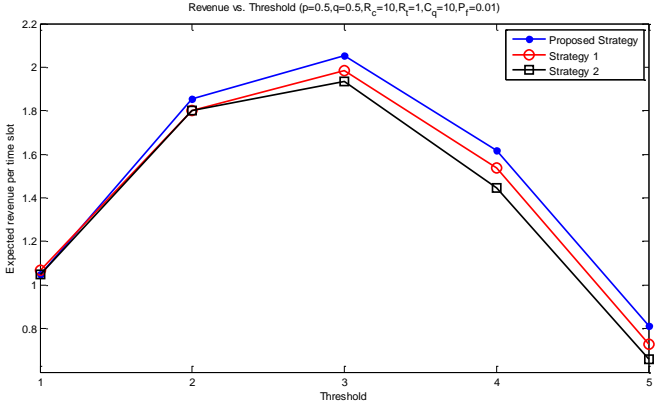}\\
\caption{Revenue versus threshold of three different strategies
($J=5, D_{\max}=5$).}\label{fig:RevenueVSThreshold_3}
\end{figure}

\begin{figure}
\centering
\includegraphics[width=1.0\columnwidth]{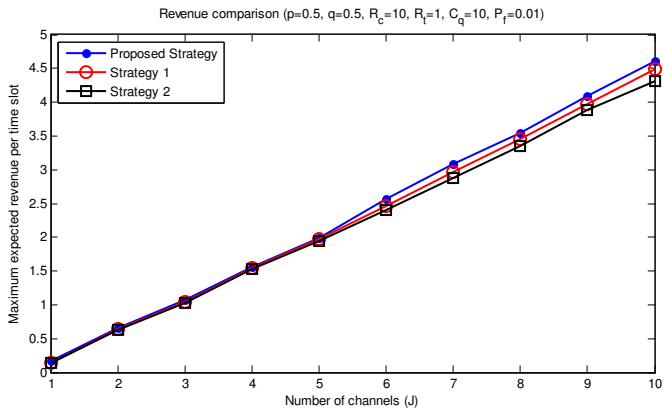}\\
\caption{Maximum expected revenue comparison with different
channels.}\label{fig:RevenueVSThreshold_max}
\end{figure}

{In Fig.~\ref{fig:RevenueVSThreshold_3} and Fig.~\ref{fig:RevenueVSThreshold_max}, we compare the proposed channel
allocation policy and the two benchmark policies with different total number of channels. }
All three policies follow the same
threshold-based admission control policies. From these figures, we
observe that the proposed largest-delay-first policy is no worse
than the other two under all choices of parameters.


\subsection{Rollout Control Policy}\label{Sec:DP.optimal}

The heuristic algorithm proposed in Section \ref{Sec:DP.solution}
can be further improved by the rollout algorithm. The general
background of the rollout algorithm is in Appendix
\ref{Sec:DP.rollout}. In this subsection, based on the analysis of
the heuristic control policy, we propose a simplified rollout
algorithm (\emph{rollout control policy}) to further improve the
performance.

{Consider two different user states $\boldsymbol{\omega_e^{(1)}}$ and
$\boldsymbol{\omega_e^{(2)}}$ that have the same number of secondary
users. If it is possible for transit from state $\boldsymbol{\omega_e^{(1)}}$ to state $\boldsymbol{\omega_e^{(2)}}$ under a particular
channel condition without admitting any new user,
then obviously the total
time delay of $\boldsymbol{\omega_e^{(1)}}$ summed over all users
must be less than that of $\boldsymbol{\omega_e^{(2)}}$
(as each user either has the same delay or a larger delay during the transition).
We give the following definitions:
\begin{definition}[User State Comparison]\label{def:UserState}
Consider two different user states $\boldsymbol{\omega_e^{(1)}}$ and
$\boldsymbol{\omega_e^{(2)}}$ that have the same number of secondary
users. If it is possible to transit from state $\boldsymbol{\omega_e^{(1)}}$ to state $\boldsymbol{\omega_e^{(2)}}$ under a particular
channel condition without admitting any new user,
then $\boldsymbol{\omega_e^{(1)}}$ is \emph{better} than $\boldsymbol{\omega_e^{(2)}}$, denoted
$\boldsymbol{\omega_e^{(1)}} \gtrdot \boldsymbol{\omega_e^{(2)}}$.
\end{definition}
\begin{definition}[Quality of Channel State]\label{def:ChannelState}
Consider a user state $\boldsymbol{\omega_e^{(1)}}$ and a channel
state $m$. The channel state $m$ is $\boldsymbol{B}$ (Bad) for the
user state $\boldsymbol{\omega_e^{(1)}}$ if and only if $m$ is less
than the total number of users in $\boldsymbol{\omega_e^{(1)}}$.
Otherwise, the channel state $m$ is $\boldsymbol{G}$ (Good) for the
user state $\boldsymbol{\omega_e^{(1)}}$.
\end{definition}}

Now consider a heuristic control policy with the admission control
threshold $N_{th}$ and largest-delay-first channel allocation
mechanism. Under this policy, we can divide the infinite-horizon
process into infinite number of segments separated by the time slots
in which there is at least one user leaving the system (normal
completion or forced termination). Then we can define a new average
revenue $\bar{g}(N_{th},\theta)$ and its expectation
$\bar{G}(N_{th},\theta)$) over each segment. Due to the
threshold-based admission control, we will only admit new users in
the first slot of a segment.

\begin{definition}[Average Revenue and Expected Average Revenue]\label{def:AveRev}\label{def:Expectation}
If the network state is $\theta$ at the beginning of time slot $k$,
and at least one user leaves the system for the first time (normal
completion or forced termination) in time slot $k+\delta$, we define
the average revenue over the period $[k,k+\delta]$ as
\begin{equation}\label{equ:InstantRevenue}
\bar{g}(N_{th},\theta)\!=\!\frac{n_c(N_{th},\theta)}{\delta+1}R_c\!-\!\frac{n_d(N_{th},\theta)}{\delta+1}C_q\!+\!N_{th}R_t,
\end{equation}
where $n_c(N_{th},\theta)$ is number of users completing connections
normally in time slot $k+\delta$, and {$n_d(N_{th},\theta)$} is number of
users being forced to terminate in time slot $k+\delta$.
The expected average revenue is denoted as
\begin{equation}\label{equ:G_bar}
\begin{split}
\bar{G}(N_{th},\theta)&=E\{\bar{g}(N_{th},\theta)\}\\
&=\!N_c(N_{th},\theta)R_c\!-\!N_d(N_{th},\theta)C_q\!+\!N_{th}R_t,
\end{split}
\end{equation}
where
$N_c(N_{th},\theta)=E\left\{\frac{n_c(N_{th},\theta)}{\delta+1}\right\}$
and
$N_d(N_{th},\theta)=E\left\{\frac{n_d(N_{th},\theta)}{\delta+1}\right\}.$
\end{definition}

The \emph{expected revenue} in Definition \ref{def:AveRev} is
different from the \emph{instant revenue} in
(\ref{equ:RevenueFunction-3}). The expected revenue is defined under
a very special case, where no new users are admitted except in the
first time slot and no users leave the network except in the last
time slot of the interval. The instant revenue defined in
(\ref{equ:RevenueFunction-3}) is the revenue for a generic time
slot. Furthermore, $\bar{G}(N_{th},\theta)$ represents the expected
average revenue per time slot when maintaining a fixed number of
users until someone leaves. Although the precise value of
$\bar{G}(N_{th},\theta)$ is hard to compute explicitly,  we have the
following result as a corollary of Proposition
\ref{pro:OptimalChannelAllocation}.

\begin{proposition}\label{pro:G}
Given any fixed $N_{th}$ and $\theta$, the largest-delay-first channel
allocation policy achieves the maximum
$\bar{G}(N_{th},\theta)$.
\end{proposition}

Based on Proposition \ref{pro:G}, we will still use the
largest-delay-first strategy channel allocation. The key remaining
issue is how to improve the admission control policy. Next we
characterize the properties of the largest-delay-first channel
allocation policy (the expected average revenue
$\bar{G}(N_{th},\theta)$ in the heuristic control policy) in several
lemmas, which enable us to design a better heuristic algorithm for
the admission control part.

According to the lemmas given in Appendix \ref{App:InterLemmas}, we
can characterize $\bar{G}(N_{th},\theta)$ as follows.
\begin{proposition}\label{Lemma:G}
$\bar{G}(N_{th},\theta)$ is a concave function of $N_{th}$.
\end{proposition}
\begin{proof}
The second order derivative of $\bar{G}(N_{th},\theta)$ in
terms of $N_{th}$ is
\begin{equation}\label{equ:Derivative2}
\bar{G}^{\prime\prime}(N_{th},\theta)=N_c^{\prime\prime}(N_{th},\theta)R_c-N_d^{\prime\prime}(N_{th},\theta)C_q,
\end{equation}
where $N_c^{\prime\prime}(N_{th},\theta)<0$ and
$N_d^{\prime\prime}(N_{th},\theta)>0$ based on Lemma \ref{Lemma:Pc}
and Lemma \ref{Lemma:Pd} in Appendix \ref{App:InterLemmas}. Thus we
have $\bar{G}^{\prime\prime}(N_{th},\theta)<0$, \ie
$\bar{G}(N_{th},\theta)$ is a concave function of $N_{th}$.
\end{proof}

\begin{figure}
\centering
\includegraphics[width=1.0\columnwidth]{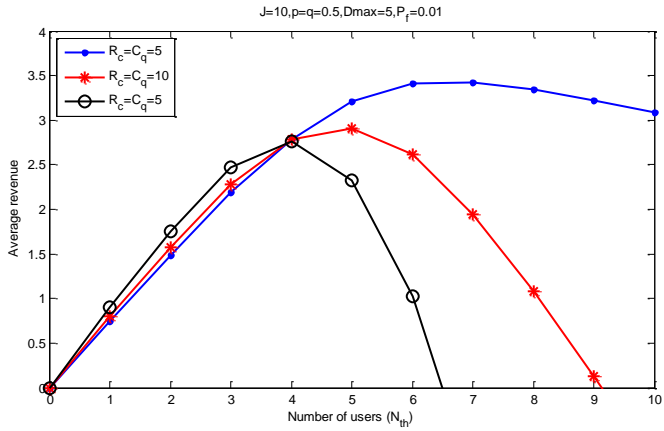}\\
\caption{The values of $\bar{G}(N_{th},\theta)$ versus $N_{th}$
corresponding to different values of $R_c$ and $C_q$ when $J=10$,
$p=0.5$, $q=0.5$, $D_{max}=5$, $P_f=0.01$, $R_t=0.7$ and
$\theta=\left\{m,[0,0,0,0,0,0]\right\}$.}\label{fig:GRcCq}
\end{figure}

Figure \ref{fig:GRcCq} plots $\bar{G}(N_{th},\theta)$ versus
$N_{th}$ with fixed $\theta=\left\{m,[0,0,0,0,0,0]\right\}$ and
different values of $R_c$ and $C_q$.

Now we are ready to discuss the heuristic admission control policy.
Given a state $\theta=(m,\boldsymbol{\omega_e})$, the admission
control decision can be either \emph{maintaining} or
\emph{searching}, depending on the relationship between the channel
state component $m$ and user state component
$\boldsymbol{\omega_e}$. More precisely, if $m$ is $\boldsymbol{B}$
(Bad) for $\boldsymbol{\omega_e}$, the network coordinator will
maintain the current user population and do not admit any new user
(\ie maintaining). This is because the network resource is not
enough to support the current users, and admitting new users will
make the situation worse. If $m$ is $\boldsymbol{G}$ (Good) for
$\boldsymbol{\omega_e}$, the network coordinator first searches for
the value of $N_{th}^{\ast}$ that achieves
$\max_{N_{th}}\bar{G}(N_{th},\theta)$ (\ie searching), and then
admits the number of users equal to the difference between
$N_{th}^{\ast}$ and the current users in the network. Proposition
\ref{Lemma:G} shows that $\bar{G}(N_{th},\theta)$ has a unique
maximizer $N_{th}^{\ast}$ (with a fixed state $\theta$), and implies
a simple stopping rule for the numerical search. If we have
$\bar{G}(N_{th}^{\prime}-1,\theta)\leq
\bar{G}(N_{th}^{\prime},\theta)$ and
$\bar{G}(N_{th}^{\prime},\theta) \geq
\bar{G}(N_{th}^{\prime}+1,\theta)$, then
$N_{th}^{\ast}=N_{th}^{\prime}$.

The heuristic admission control introduced above is a rollout
control policy based on the theory in Appendix \ref{Sec:DP.rollout}.
More specifically, the value of
$\max_{N_{th}}\bar{G}(N_{th},\theta)$ computed in the searching step
is the \emph{cost-to-go} starting from a state $\theta$. As
Proposition \ref{Lemma:G} shows that this is a concave maximization
problem, we can use several well-known numerical methods to achieve
this. One possibility is the gradient decent method, which has a
linear convergence rate as shown in \cite{Boyd2004.Convex}. More
precisely, the maximum number of convergence of the gradient decent
method is proportional to $\log{(\bar{G}(N_{th}^{initial},\theta) -
\bar{G}(N_{th}^{optimal},\theta))}/{\epsilon}$, where $\epsilon$ is
the stopping criterion. Since the precise value of
$\bar{G}(N_{th},\theta)$ is hard to compute with a low complexity,
we will use an approximation $\tilde{G}(N_{th},\theta)$ instead  in
the searching step. In this paper, we use an on-line computation
(simulation) to get $\tilde{G}(N_{th},\theta)$. Mover specifically,
for each choice of $(N_{th},\theta)$, we can obtain the value of
$\bar{g}((N_{th},\theta))$ as in (\ref{equ:InstantRevenue}) for each
particular simulation, and take the average over many simulations to
obtain an approximation $\tilde{G}(N_{th},\theta)$. The memory
requirements are proportional to the expected length of the segments
separated by the time slots in which there is at least one user
leaving the system (normal completion or forced termination)


\subsection{Revenue Boundary}

In this subsection, we will compare the performance of two heuristic
policies that we have proposed. Before that, we will establish an
upper-bound of the revenue achievable under any control policy
(heuristic or optimal). We call the bound the \emph{revenue
boundary}.

We first prove the following property of the expected average
revenue $\bar{G}(N_{th},\theta)$.
\begin{proposition}\label{Lemma:G2}
For a fixed number of users $N_{th}$, if there are two states
$\theta_1=\left\{m,\boldsymbol{\omega_e^{(1)}}\right\}$ and
$\theta_2=\left\{m,\boldsymbol{\omega_e^{(2)}}\right\}$ such that
$\boldsymbol{\omega_e^{(1)}} \gtrdot \boldsymbol{\omega_e^{(2)}}$,
we have $\bar{G}(N_{th},\theta_1)>\bar{G}(N_{th},\theta_2)$.
\end{proposition}
\begin{proof}
According to Lemma \ref{Lemma:Pd2} in Appendix
\ref{App:InterLemmas}, we have
$N_c(N_{th},\theta_1)>N_c(N_{th},\theta_2)$ and
$N_d(N_{th},\theta_1)<N_d(N_{th},\theta_2)$. By substituting them
into (\ref{equ:G_bar}), we get
$\bar{G}(N_{th},\theta_1)>\bar{G}(N_{th},\theta_2)$.
\end{proof}

Then we can characterize the revenue boundary.
\begin{proposition}\label{pro:Boundary}
Consider a network state
$\bar{\theta}=\left\{m,[0,0,0,\cdots]\right\}$, where there are $m$
available channels. The maximum expected revenue per time slot
achieved by any policy, denoted by $G_{max}$, satisfies
$G_{max}<\max_{m,N_{th}}\left\{\bar{G}(N_{th},\bar{\theta})\right\},$
where $N_{th}$ is an admission control threshold.
\end{proposition}

\begin{proof}
Assume $\hat{\theta}=\left\{m,
\boldsymbol{\omega_e^{\theta}}\right\}$ and $\eta = \left\{m,
\boldsymbol{\omega_e^{\eta}}\right\}$ are two network states with
$m$ available channels and $N_{th}$ users, where
$\boldsymbol{\omega_e^{\theta}}=[N_{th},0,0,\cdots]$. If
$\boldsymbol{\omega_e^{\eta}} \neq \boldsymbol{\omega_e^{\theta}}$,
we have $\boldsymbol{\omega_e^{\theta}} \gtrdot
\boldsymbol{\omega_e^{\eta}}$. From Proposition \ref{Lemma:G2}, we
get $\bar{G}(N_{th},\hat{\theta})>\bar{G}(N_{th},\eta)$. In
addition, after the control decision in the first time slot,
$N_{th}$ new secondary users are admitted in the case of
$\overline{\theta}$ (since there are originally no users in the
system), and no new secondary user is admitted in the case of
$\hat{\theta}$ (since there are already $N_{th}$ users with zero
accumulative delay in the system). Thus, after the first time slot,
we achieve the same state in both cases. In the following time
slots, the expected changes of the two cases are thus the same.
Therefore, according to the definition of $\overline{G}$ in
Definition \ref{def:AveRev}, we have $\bar{G}(N_{th},\bar{\theta}) =
\bar{G}(N_{th},\hat{\theta})$. Therefore, the cost-to-go we compute
in the search step is never larger than
$\max_{m,N_{th}}\left\{\bar{G}(N_{th},\bar{\theta})\right\}$. As the
optimal policy can be viewed as a special case of the rollout policy
by using the optimal policy as the \emph{base policy}, it follows
that the expected revenue per time slot of any policy (including the
optimal one) is less than
$\max_{m,N_{th}}\left\{\bar{G}(N_{th},\bar{\theta})\right\}$. %
\end{proof}

\begin{figure}
\centering
\includegraphics[width=1.0\columnwidth]{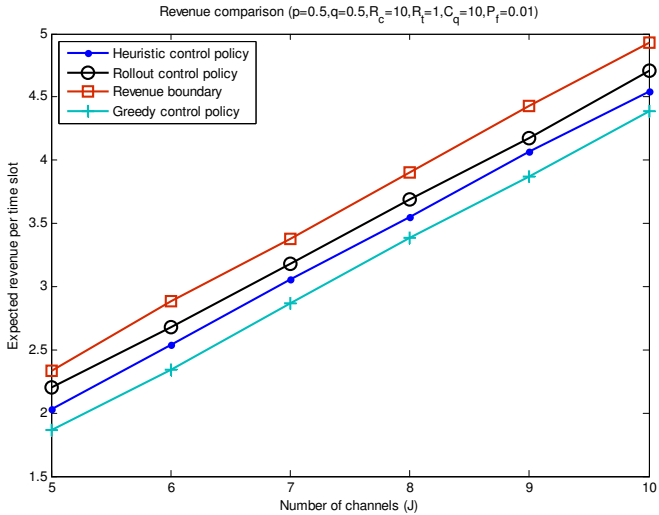}\\
\caption{Expected revenue comparison between the greedy control
policy, the heuristic control policy, the rollout control policy,
and the revenue boundary when $D_{max}=5$, $p=0.5$, $q=0.5$,
$R_c=10$, $R_t=1$, $C_q=10$, $P_f=0.01$ and $J \in
[5,10]$.}\label{fig:RComp}
\end{figure}

{In Section~\ref{Sec:SystemModel}, we have assumed perfect spectrum
sensing.  Under this assumption, the control policy of the
throughput maximization problem studied in
\cite{ZhouLiLiWaSo2010,UrN09} can be simplified into admitting
secondary users to make full use of the available channels in each
time slot, which we call \emph{greedy admission control} in this
paper. Such greedy admission control policy will admit new users
whenever possible such that the total active users in a time slot
equals to the number of available channels. Comparing with our
proposed policy, this greedy policy is more aggressive and does not
consider channel availabilities in the future, and thus will lead to
a larger number of forced dropped users. We have plotted the
expected revenue of the greedy admission control in
Fig.~\ref{fig:RComp}, with the comparison with our proposed
admission control and the revenue boundary. We can see that even the
performance of our proposed heuristic control policy is better than
that of the greedy control policy. The heuristic control policy
(with the threshold-based admission control) is simple but
effective, while the rollout algorithm achieves a slightly better
performance but with a much higher computational complexity.} The
actual performance gap between the proposed algorithms and the
optimal policy could be even smaller, as the revenue boundary in
Proposition \ref{pro:Boundary} may not be very tight.


\section{Conclusions}\label{Sec:Conclusion}

Supporting QoS over cognitive radio networks is very challenging,
mainly due to the uncertainty of available communication resources.
As one further step towards understanding this under-explored yet
practically important research area, we considered supporting delay
sensitive traffic in cognitive radio networks. The key is to jointly
optimize admission control and channel  allocation, in order to
balance the number of concurrent sessions and the QoS of each
session. We formulated the problem as an infinite-horizon Markov
decision process problem, and proved that the optimal average
revenue is independent of the initial system state. Then we
transformed the original problem into a stochastic shortest path
problem, and proved that the Bellman's equation converged to the
optimal policy. Furthermore, we proposed a heuristic control policy
and proved that the largest-delay-first strategy is optimal given
threshold-based admission control. We further proposed a rollout
algorithm that improves upon the heuristic algorithm by doing
dynamic admission control. By comparing with a revenue bound, we
show that both of our proposed algorithms achieve close-to-optimal
performance.


\appendix

\subsection{Rollout Algorithm}\label{Sec:DP.rollout}

For convenience, we consider the finite-horizon stochastic shortest
path problem as a  discrete-time dynamic system
\begin{equation}\label{equ:DynamicProgramming}
x_{k+1}=f(x_k,u(x_k),\zeta_k),~k=0,1,\cdots.
\end{equation}
According to definitions in Section \ref{Sect:ProblemFormualtion},
$x_k$ is the state (belonging to the state space $\mathcal{S}$) at
time slot $k$, $u(x_k)$ is the control selected from the control
space $\mathcal{U}$ at time $k$, $\zeta_k$ is a random disturbance
caused by the activities of the users at time $k$, and $f$ is the
state transition function. We focus on an $N$-stage horizon problem
with a terminal cost $g(x_N)$ that depends on the terminal state
$x_N$. We define the \emph{cost-to-go} of a policy
$\boldsymbol{\mu}=\{u(x_0),u(x_1),\cdots,u(x_{N-1})\}$ starting from
a state $x_k$ at time slot $k$ as
\begin{equation}\label{equ:cost2go}
J_k^{\boldsymbol{\mu}}(x_k)=E\left\{A^\ast-g(x_N)+
\sum_{i=k}^{N-1}{\{A^\ast-g(x_i,u(x_i))\}}\right\}.
\end{equation}
The optimal cost-to-go starting from a state $x_k$ in time slot $k$
is $J_k(x_k)=\inf_{\boldsymbol{\mu}}{J_k^{\boldsymbol{\mu}}(x_k)}$,
and it satisfies the following recursive relationship
\begin{equation}\label{equ:DPRecursion}
\begin{split}
J_k(x_k) \!\! = \!\!\inf_{\mu_k \in
\mathcal{U}}{E\left\{A^\ast\!\!-\!\!g(x_k,u(x_k))\!\!+\!\!
J_{k+1}(f(x_k,u(x_k),w_k))\right\}},
\end{split}
\end{equation}
with $k=0,1,\cdots,N-1$ and the initial condition is
$J_N(x_N)=A^\ast-g(x_N)$. We can also extend the definitions to
infinite-horizon problems with minor modifications.

An optimal policy could be obtained by calculating the optimal
cost-to-go functions $J_k$. But it is prohibitively time-consuming
for our problem. To reduce the computation complexity, we can adopt
the rollout algorithm by replacing the optimal cost-to-go function
$J_{k+1}$ in (\ref{equ:DPRecursion}) with an approximation
$\tilde{J}_{k+1}$.

In the rollout algorithm, some known heuristic or suboptimal policy
$\boldsymbol{\mu}$, called the \emph{base policy}, will be used to
calculate the approximating function $\tilde{J}_{k+1}$.  The values
of the approximate cost-to-go $\tilde{J}_{k+1}$ may be computed in a
number of ways: by a closed-form expression, by an approximate
off-line computation, or by an on-line computation.  The improved
policy is called the \emph{rollout policy} based on
$\boldsymbol{\mu}$. It is a one-step lookahead policy (by using
(\ref{equ:DPRecursion}) once), where we approximate the optimal
cost-to-go on the right hand side of (\ref{equ:DPRecursion}) by the
cost-to-go of the base policy. The more detailed description of the
rollout algorithm can be found in references \cite{Ber05-b,Ber07-b}.

\subsection{Several Lemmas for Proving Propositions \ref{Lemma:G} and \ref{Lemma:G2}}\label{App:InterLemmas}

After defining the expected revenue $\bar{g}(N_{th},\theta)$ and the
expected average revenue $\bar{G}(N_{th},\theta)$ in Definition
\ref{def:Expectation}, we give the following intermediate lemmas to
help to illustrate the properties of $\bar{G}(N_{th},\theta)$ in
terms of the first and second order derivatives.

\begin{lemma}\label{Lemma:Pc}
For a fixed state $\theta$, $N_c(N_{th},\theta)$ is a non-decreasing
and concave function of the number of users $N_{th}$.
\end{lemma}
\begin{proof}\label{proof:Lemma.Pc}
Recall that all the users have the same completion probability $P_f$
when they are actively served. Thus we have $N_c(N_{th}+1,\theta)
\geq N_c(N_{th},\theta)$, as having one more user means that it is
possible to actively serve one more user and thus have one more
normal session completion. Furthermore, we assume that under the
same channel condition and over a period of time slots, the
incremental number of served users per time slot is $\Delta_1$ when
the number of users changes from $N_{th}-1$ to $N_{th}$. Then
$\Delta_2$, the incremental number of served users per time slot
from $N_{th}$ to $N_{th}+1$, should be no bigger than $\Delta_1$.
This is because if $N_{th}+1$ users can be allocated available
channels, $N_{th}$ users could be allocated available channels in
the same time slot. Therefore, we have
$N_c(N_{th}+1,\theta)-N_c(N_{th},\theta) \leq
N_c(N_{th},\theta)-N_c(N_{th}-1,\theta),$ which means
$N_c(N_{th},\theta)$ is a non-decreasing and concave function of
$N_{th}$.
\end{proof}

\begin{lemma}\label{Lemma:Pd}
For a fixed state $\theta$, $N_d(N_{th},\theta)$ is a non-decreasing
and convex function of the number of users $N_{th}$.
\end{lemma}

{\begin{proof} \label{proof:Lemma.Pd} Having one more admitted user
means that a higher probability of a forced termination, \ie
$N_d(N_{th}+1,\theta) \geq N_d(N_{th},\theta)$. Under the
largest-delay-first channel allocation policy, define $\Delta_1 =
N_d(N_{th},\theta)-N_d(N_{th}-1,\theta)$ and the additional user as
$U_{ser}$, and $\Delta_2 = N_d(N_{th}+1,\theta)-N_d(N_{th},\theta)$.
For discussion convenience, we call the system with $N_{th}-1$ users
as \emph{Case 1}, the system with $N_{th}$ users as \emph{Case 2},
and the system with $N_{th}+1$ users as \emph{Case 3}. In Case 2, we
divide users into two parts: $U_{ser}$ and other $N_{th}-1$ users.
In Case 3, we also divide users into two parts: $U_{ser}$ and other
$N_{th}$ users. Then we define
$N_d(N_{th}+1,\theta)=N_d^\prime(N_{th},\theta)+N_d^3(U_{ser},\theta)$
and
$N_d(N_{th},\theta)=N_d^\prime(N_{th}-1,\theta)+N_d^2(U_{ser},\theta)$.
Here $N_d^3(U_{ser},\theta)$ and $N_d^\prime(N_{th},\theta)$
represent the corresponding parts of $N_d(N_{th}+1,\theta)$ caused
by the forced termination of $U_{ser}$ and other users in Case 3,
respectively; $N_d^2(U_{ser},\theta)$ and
$N_d^\prime(N_{th}-1,\theta)$ represent the corresponding parts of
$N_d(N_{th},\theta)$ caused by the forced termination of $U_{ser}$
and other users in Case 2, respectively. On this basis, we further
define $\Delta_2 = \Delta_2^\prime + \Delta_2^{\prime\prime},$ where
$\Delta_2^\prime =
N_d^\prime(N_{th},\theta)-N_d^\prime(N_{th}-1,\theta)$ and
$\Delta_2^{\prime\prime} =
N_d^3(U_{ser},\theta)-N_d^2(U_{ser},\theta).$

In Case 2 and Case 3, we now exclude the user $U_{ser}$ from the
system and assume the channels allocated to $U_{ser}$ are occupied
by primary users. Then we can have the above expression of
$\Delta_2^\prime$ to illustrate the effect of the increased user
$U_{ser}$ from $N_{th}-1$ to $N_{th}$. Comparing $\Delta_2^\prime =
N_d^\prime(N_{th},\theta)-N_d^\prime(N_{th}-1,\theta)$ with
$\Delta_1 = N_d(N_{th},\theta)-N_d(N_{th}-1)$, the difference is
that in any time slot (on any sample path), the channel state of
$\Delta_2^\prime $ is always no better than that of the $\Delta_1 $
case (as the extra user $U_{ser}$ may occupy an available channel).
Therefore, in terms of the expected number of users forced to leave
the system per time slot, the effect of the increased user to
$\Delta_2^\prime$  is larger than that to $\Delta_1 $. This leads to
$\Delta_2^\prime \geq \Delta_1$. Moreover, considering $U_{ser}$
from Case 2 to Case 3, we have $\Delta_2^{\prime\prime} \geq 0$
under the largest-delay-first policy. From the above analysis, we
get $\Delta_2\geq\Delta_1$, \ie
$N_d(N_{th}+1,\theta)-N_d(N_{th},\theta) \geq
N_d(N_{th},\theta)-N_d(N_{th}-1,\theta),$ which means
$N_d(N_{th},\theta)$ is a non-decreasing and convex function of
$N_{th}$ \cite{Fox66}.
\end{proof}}

\begin{lemma}\label{Lemma:Pd2}
For a fixed number of users $N_{th}$, if there are two states
$\theta_1=\left\{m,\boldsymbol{\omega_e^{(1)}}\right\}$ and
$\theta_2=\left\{m,\boldsymbol{\omega_e^{(2)}}\right\}$ such that
$\boldsymbol{\omega_e^{(1)}} \gtrdot \boldsymbol{\omega_e^{(2)}}$,
we have $N_c(N_{th},\theta_1)>N_c(N_{th},\theta_2)$ and
$N_d(N_{th},\theta_1)<N_d(N_{th},\theta_2)$.
\end{lemma}

{\begin{proof}
The lemma directly follows the definitions of
$\boldsymbol{\omega_e^{(1)}} \gtrdot \boldsymbol{\omega_e^{(2)}}$ in
Definition \ref{def:UserState} and $N_c(N_{th},\theta)$,
$N_d(N_{th},\theta)$ in Definition \ref{def:Expectation}. If
$\boldsymbol{\omega_e^{(1)}} \gtrdot \boldsymbol{\omega_e^{(2)}}$,
the user state $\boldsymbol{\omega_e^{(1)}}$ can reach the user
state $\boldsymbol{\omega_e^{(2)}}$ under a proper channel condition
and a control policy. Consider two systems with the initial states
$\theta_1$ and $\theta_2$, respectively, and follow the same channel
conditions over time and the same control policy. When a user is
forced to leave the system (completes the connection, respectively)
with $\theta_1$, in the system with $\theta_2$, there must be a user
that is forced to leave (completes the connection or is forced to
leave, respectively) in the current or an earlier time slot.
Therefore, we get $N_c(N_{th},\theta_1)>N_c(N_{th},\theta_2)$ and
$N_d(N_{th},\theta_1)<N_d(N_{th},\theta_2)$ based on the definitions
of $N_c(N_{th},\theta)$ and $N_d(N_{th},\theta)$.
\end{proof}}


\bibliographystyle{IEEEtran}
\bibliography{TWC}

\begin{biography}[{\includegraphics[width=1in,height=1.25in,clip,keepaspectratio]{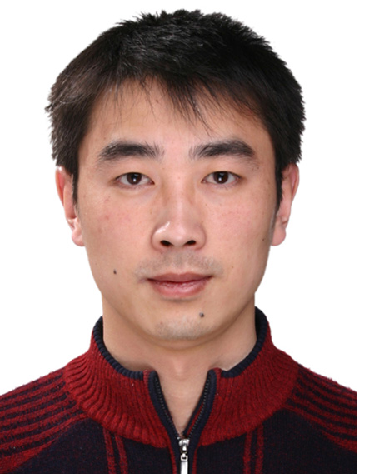}}]{Feng
Wang} received B.S. in Electronic Information Engineering from
Shandong University (Jinan, Shandong, P.R.China) in 2005 and Ph.D.
in Communication and Information System from Peking University
(Beijing, P.R.China) in 2011. He visited the Chinese University of
Hong Kong as a Research Assistant between July to December, 2009. He
is currently an engineer in the Beijing Space Technology Development
and Test Center, China Academy of Space Technology, Beijing,
P.R.China. His current research interests include resource
allocation, cognitive radio and wireless sensor networks.
\end{biography}

\begin{biography}[{\includegraphics[width=1in,height=1.25in,clip,keepaspectratio]{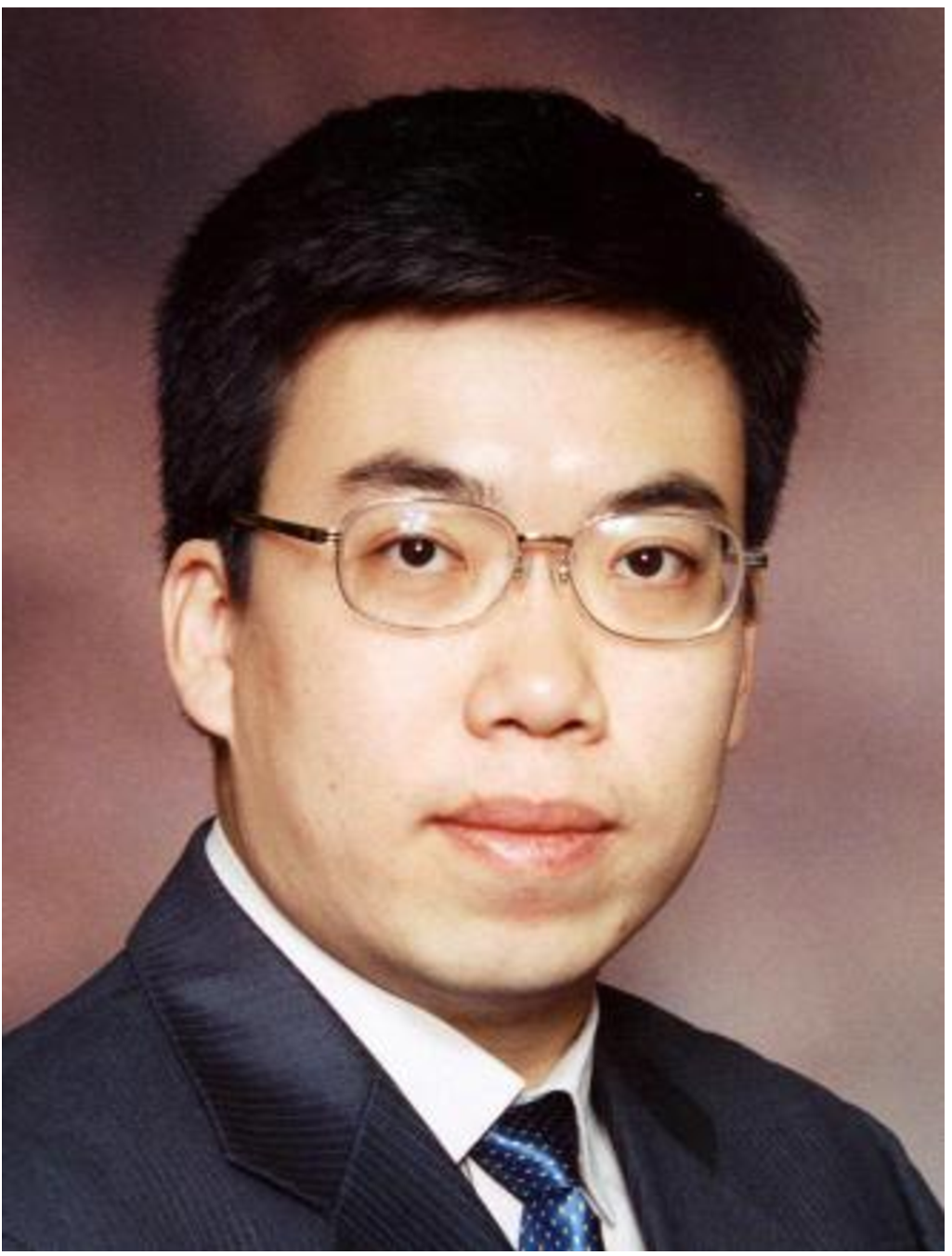}}]{Jianwei Huang}
(S'01-M'06-SM'11) is an Assistant Professor in the Department of
Information Engineering at the Chinese University of Hong Kong. He
received B.S. in Electrical Engineering from Southeast University
(Nanjing, Jiangsu, China) in 2000, M.S. and Ph.D. in Electrical and
Computer Engineering from Northwestern University (Evanston, IL,
USA) in 2003 and 2005, respectively. He worked as a Postdoc Research
Associate in the Department of Electrical Engineering at Princeton
University during 2005-2007. He was a visiting scholar in the School
of Computer and Communication Sciences at \'{E}cole Polytechnique
F\'{e}d\'{e}rale De Lausanne (EPFL) during the Summer Research
Institute in June 2009, and a visiting scholar in the Department of
Electrical Engineering and Computer Sciences at University of
California-Berkeley in August 2010.

Dr. Huang currently leads the Network Communications and Economics
Lab (ncel.ie.cuhk.edu.hk), with main research focus on nonlinear
optimization and game theoretical analysis of communication
networks, especially on network economics, cognitive radio networks,
and smart grid. He is the recipient of the IEEE Marconi Prize Paper
Award in Wireless Communications in 2011, the International
Conference on Wireless Internet Best Paper Award 2011, the IEEE
GLOBECOM Best Paper Award in 2010, the IEEE ComSoc Asia-Pacific
Outstanding Young Researcher Award in 2009, Asia-Pacific Conference
on Communications Best Paper Award in 2009, and Walter P. Murphy
Fellowship at Northwestern University in 2001.

Dr. Huang has served as Editor of IEEE Journal on Selected Areas in
Communications - Cognitive Radio Series, Editor of IEEE Transactions
on Wireless Communications, Guest Editor of IEEE Journal on Selected
Areas in Communications special issue on ``Economics of
Communication Networks and Systems'', Lead Guest Editor of IEEE
Journal of Selected Areas in Communications special issue on ``Game
Theory in Communication Systems'', Lead Guest Editor of IEEE
Communications Magazine Feature Topic on ``Communications Network
Economics'', and Guest Editor of several other journals including
(Wiley) Wireless Communications and Mobile Computing, Journal of
Advances in Multimedia, and Journal of Communications.

Dr. Huang has served as Vice Chair of IEEE MMTC (Multimedia
Communications Technical Committee) (2010-2012), Director of IEEE
MMTC E-letter (2010), the TPC Co-Chair of IEEE WiOpt (International
Symposium on Modeling and Optimization in Mobile, Ad Hoc, and
Wireless Networks) 2012, the Publicity Co-Chair of IEEE
Communications Theory Workshop 2012, the TPC Co-Chair of IEEE ICCC
Communication Theory and Security Symposium 2012, the Student
Activities Co-Chair of IEEE WiOpt 2011, the TPC Co-Chair of IEEE
GlOBECOM Wireless Communications Symposium 2010, the TPC Co-Chair of
IWCMC (the International Wireless Communications and Mobile
Computing) Mobile Computing Symposium 2010, and the TPC Co-Chair of
GameNets (the International Conference on Game Theory for Networks)
2009. He is also TPC member of leading conferences such as INFOCOM,
MobiHoc, ICC, GLBOECOM, DySPAN, WiOpt, NetEcon, and WCNC. He is a
senior member of the IEEE.
\end{biography}

\begin{biography}[{\includegraphics[width=1in,height=1.25in,clip,keepaspectratio]{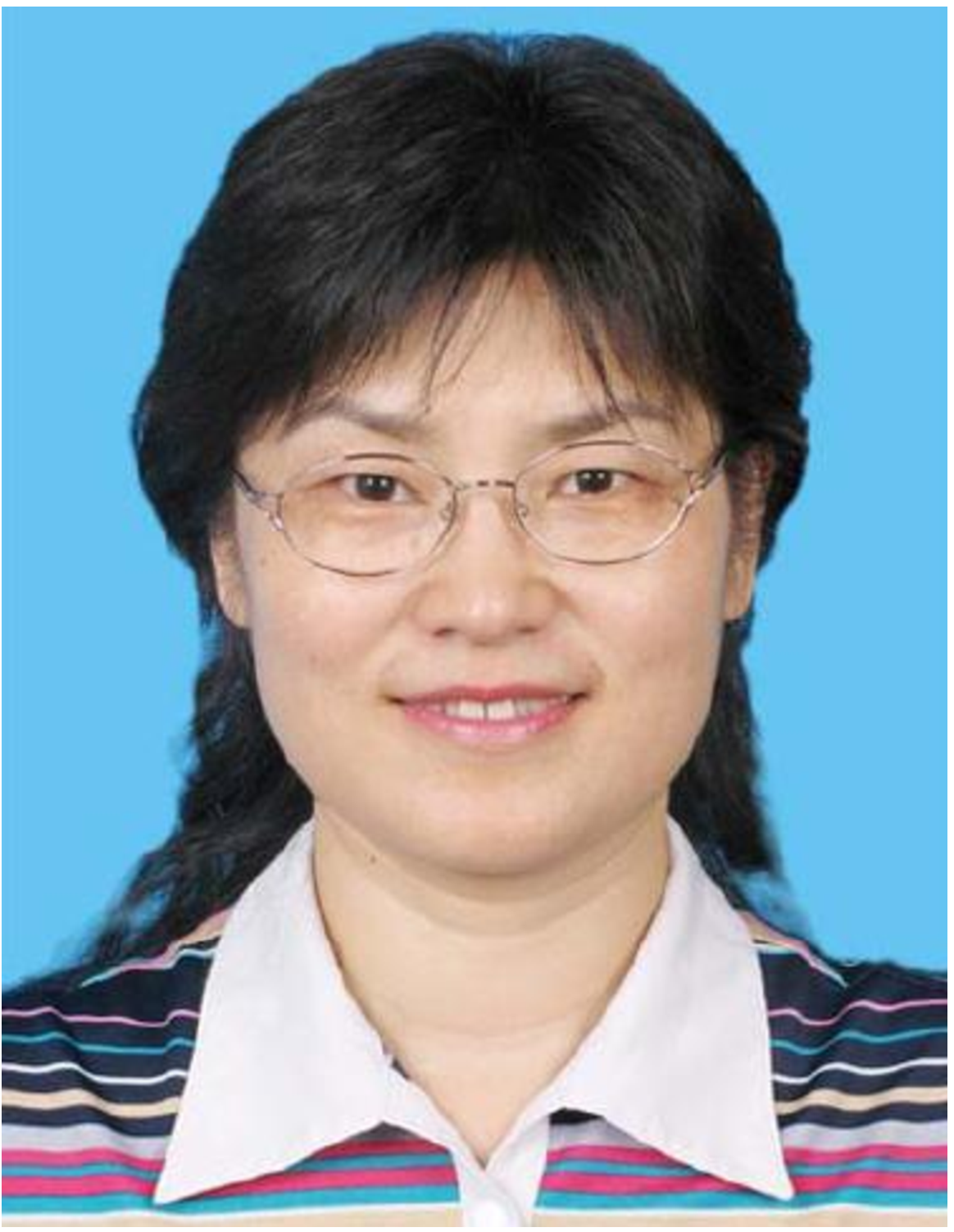}}]{Yuping Zhao}
received the B.S. and M.S. degrees in electrical engineering from
Northern Jiaotong University, Beijing, P.R.China, in 1983 and 1986,
respectively. She received the Ph.D. and Doctor of Science degrees
in wireless communications from Helsinki University of Technology,
Helsinki, Finland, in 1997 and 1999, respectively. She was a System
Engineer for telecommunication companies in China and Japan. She
worked as a research engineer at the Helsinki University of
Technology, Helsinki, Finland, and at the Nokia Research Center in
the field of radio resource management for wireless mobile
communication networks. Currently, she is a professor in the State
Key Laboratory of Advanced Optical Communication Systems \&
Networks, School of Electronics Engineering and Computer Science,
Peking University, Beijing, P.R.China. Her research interests
include the areas of wireless communications and corresponding
signal processing, especially for OFDM, UWB and MIMO systems,
cooperative networks, cognitive radio, and wireless sensor networks.
\end{biography}

\end{document}